\documentclass[10pt]{article}

\usepackage[utf8]{inputenc}
\usepackage[T1]{fontenc}

\usepackage{epsf}
\usepackage{amsmath}

\allowdisplaybreaks

\usepackage[showframe=false]{geometry}
\usepackage{changepage}

\usepackage{epsfig}
\usepackage{amssymb}

\usepackage{amsthm}
\usepackage{setspace}
\usepackage{cite}
\usepackage{mcite}

\usepackage{algorithmic}  
\usepackage{algorithm}

\usepackage{shadow}
\usepackage{fancybox}
\usepackage{fancyhdr}

\usepackage{color}
\usepackage[usenames,dvipsnames,svgnames,table]{xcolor}
\newcommand{\bl}[1]{\textcolor{blue}{#1}}
\newcommand{\red}[1]{\textcolor{red}{#1}}

\definecolor{mypurple}{rgb}{.4,.0,.5}

\usepackage[hyphens]{url}

\usepackage[colorlinks=true,
            linkcolor=black,
            urlcolor=blue,
            citecolor=purple]{hyperref}

\usepackage{breakurl}

\def\s{{\bf s}}

\def\y{{\bf y}}

\def\x{{\bf x}}

\def\x{{\mathbf x}}

\def\s{{\bf s}}

\def\x{{\bf x}}
\def\y{{\bf y}}
\def\z{{\bf z}}

\def\b{{\bf b}}

\def\d{{\bf d}}
\def\f{{\bf f}}

\def\tr{\mbox{Tr}}

\def\tr{{\rm tr}\,}

\def\cS{{\mathcal S}}

\def\be{\begin{equation}}
\def\ee{\end{equation}}
\def\ba{\left[\begin{array}}
\def\ea{\end{array}\right]}

\def\s{{\bf s}}

\def\x{{\bf x}}
\def\y{{\bf y}}
\def\z{{\bf z}}

\def\b{{\bf b}}

\def\d{{\bf d}}
\def\f{{\bf f}}

\def\1{{\bf 1}}

\def\g{{\bf g}}
\def\0{{\bf 0}}

\def\mR{{\mathbb R}}

\def\mE{{\mathbb E}}

\def\mP{{\mathbb P}}

\def\lp{\left (}
\def\rp{\right )}

\def\s{{\bf s}}

\def\y{{\bf y}}

\def\x{{\bf x}}

\def\x{{\mathbf x}}

\def\s{{\bf s}}

\def\x{{\bf x}}
\def\y{{\bf y}}
\def\z{{\bf z}}

\def\b{{\bf b}}

\def\d{{\bf d}}
\def\f{{\bf f}}

\def\tr{\mbox{Tr}}

\def\tr{{\rm tr}\,}

\def\be{\begin{equation}}
\def\ee{\end{equation}}
\def\ba{\left[\begin{array}}
\def\ea{\end{array}\right]}

\def\s{{\bf s}}

\def\x{{\bf x}}
\def\y{{\bf y}}
\def\z{{\bf z}}

\def\b{{\bf b}}

\def\d{{\bf d}}
\def\f{{\bf f}}

\def\({\left (}
\def\){\right )}

\def\1{{\bf 1}}

\def\g{{\bf g}}
\def\0{{\bf 0}}

\usepackage{xcolor}
\usepackage{color}

\definecolor{darkgreen}{rgb}{0, 0.4,0}
\newcommand{\dgr}[1]{\textcolor{darkgreen}{#1}}

\definecolor{purplebrown}{rgb}{0.5,0.1,0.6}

\definecolor{ultclupcol}{rgb}{0.1,0.5,0.5}

\definecolor{mytrycolor}{rgb}{0.5,0.7,0.2}

\definecolor{ultclupcola}{rgb}{.5,0,.5}

\definecolor{shadebrown}{rgb}{0.1,0.1,0.9}
\definecolor{lightblue}{rgb}{0.2,0,1}

\usepackage{fancybox}
\usepackage{graphicx}
\usepackage{epstopdf}
\usepackage{epsfig}
\usepackage{wrapfig}
\usepackage{subfigure}

\usepackage{xcolor}
\usepackage{tcolorbox}
\tcbuselibrary{skins}

\newtcbox{\xmybox}{on line,
arc=7pt,
before upper={\rule[-3pt]{0pt}{10pt}},boxrule=0pt,
boxsep=0pt,left=6pt,right=6pt,top=0pt,bottom=0pt,enhanced, coltext=blue, colback=white!10!yellow}

\newtcbox{\xmyboxa}{on line,
arc=7pt,
before upper={\rule[-3pt]{0pt}{10pt}},boxrule=0pt,
boxsep=0pt,left=6pt,right=6pt,top=0pt,bottom=0pt,enhanced, colback=white!10!yellow}

\newtcbox{\xmyboxb}{on line,
arc=7pt,
before upper={\rule[-3pt]{0pt}{10pt}},boxrule=1pt,colframe=darkgreen!100!blue,
boxsep=0pt,left=6pt,right=6pt,top=0pt,bottom=0pt,enhanced, colback=white!10!yellow}

\newtcbox{\xmyboxc}{on line,
arc=7pt,
before upper={\rule[-3pt]{0pt}{10pt}},boxrule=.7pt,colframe=blue!100!blue,
boxsep=0pt,left=6pt,right=6pt,top=0pt,bottom=0pt,enhanced, coltext=blue, colback=white!10!yellow}

\newtcbox{\xmytboxa}{on line,
arc=7pt,
before upper={\rule[-3pt]{0pt}{10pt}},boxrule=.0pt,colframe=pink!50!yellow,
boxsep=0pt,left=6pt,right=6pt,top=0pt,bottom=0pt,enhanced, coltext=white, colback=blue!40!red}

\newtcbox{\xmytboxb}{on line,
arc=7pt,
before upper={\rule[-3pt]{0pt}{10pt}},boxrule=.0pt,colframe=pink!50!yellow,
boxsep=0pt,left=6pt,right=6pt,top=0pt,bottom=0pt,enhanced, coltext=white, colback=white!40!green}

\makeatletter
\newcommand\subsubsubsection{\@startsection{paragraph}{4}{\z@}{-2.5ex\@plus -1ex \@minus -.25ex}{1.25ex \@plus .25ex}{\normalfont\normalsize\bfseries}}
\newcommand\subsubsubsubsection{\@startsection{subparagraph}{5}{\z@}{-2.5ex\@plus -1ex \@minus -.25ex}{1.25ex \@plus .25ex}{\normalfont\normalsize\bfseries}}
\makeatother

\newtheorem{theorem}{Theorem}

\newtheorem{lemma}{Lemma}

\begin{document}

\begin{singlespace}

\title {Capacity of the treelike sign perceptrons neural networks with one hidden layer -- RDT based upper bounds
}
\author{
\textsc{Mihailo Stojnic
\footnote{e-mail: {\tt flatoyer@gmail.com}} }}
\date{}
\maketitle

\centerline{{\bf Abstract}} \vspace*{0.1in}

We study the capacity of \emph{sign} perceptrons neural networks (SPNN) and particularly focus on 1-hidden layer \emph{treelike committee machine} (TCM) architectures. Similarly to what happens in the case of a single perceptron neuron, it turns out that, in a statistical sense, the capacity of a corresponding multilayered network architecture consisting of multiple \emph{sign} perceptrons also undergoes the so-called phase transition (PT) phenomenon. This means: (i) for certain range of system parameters (size of data, number of neurons), the network can be properly trained to accurately memorize \emph{all} elements of the input dataset; and (ii) outside the region such a training does not exist. Clearly, determining the corresponding phase transition curve that separates these regions is an extraordinary task and among the most fundamental questions related to the performance of any network. Utilizing powerful mathematical engine called Random Duality Theory (RDT), we establish a generic framework for determining the upper bounds on the 1-hidden layer TCM SPNN capacity. Moreover, we do so for \emph{any} given (odd) number of neurons. We further show that the obtained results \emph{exactly} match the replica symmetry predictions of \cite{EKTVZ92,BHS92}, thereby proving that the statistical physics based results are not only nice estimates but also mathematically rigorous bounds as well. Moreover, for $d\leq 5$, we obtain the capacity values that improve on the best known rigorous ones of \cite{MitchDurb89}, thereby establishing a first, mathematically rigorous, progress in well over 30 years.


\vspace*{0.25in} \noindent {\bf Index Terms: Multi-layer neural networks; Capacity; Random duality theory}.

\end{singlespace}

\section{Introduction}
\label{sec:intro}

As is well known, neural networks (NN) are among the most powerful algorithmic and engineering tools often used to handle many, otherwise intractable,  problems. Growing demand for understanding large datasets has skyrocketed the NN relevance within the last decade. Clearly, the practical algorithmic aspects including designs and various NN applications lead the way. A corresponding trend of high demand for the adequate theoretical justifications has been created as well. We in this paper follow such a trend and provide a strong theoretical support for \emph{precisely} determining the so-called network's \emph{memory capacity} --  one of the key fundamental features of any NN.

\subsection{Model, mathematical setup, and relevant prior work}
\label{sec:model}

We start by considering a general setup with (potentially) fully connected multilayered multi-input single-output feed-forward neural net with the number of inputs $n$ and the number of hidden layers $L-2$. Deviating slightly from the common practice and for the notational convenience, we view the inputs as the nodes of the first layer and the output as the node of the last layer. This also implies that $n$ is the number of the neurons in the first  and $1$ is the number of the nodes in the last layer (the first and/or the last layer can then be viewed through the prism of the existing literature terminology as artificial layers). Also we denote by $d_i$ the number of the nodes in layer $i$, $1\leq i\leq L$ and, clearly, set $d_1\triangleq n$ and $d_{L}=1$. The nodes from layers $i$ and $i+1$ are linearly connected by weights $W^{(i)}\in\mR^{d_{i}\times d_{i+1}}$ (to ensure the full generality, we initially assume that matrices $W^{(i)}$ are full; below, we specialize to particular cases). The network effectively operates based on the following simple mathematical formalism:
\begin{equation}\label{eq:model1}
  \x^{(i+1)}=\f^{(i)}(W^{(i)}\x^{(i)}-\b^{(i)}),
\end{equation}
where $\x^{(i)}\in\mR^{d_i}$ and $\x^{(i+1)}\in\mR^{d_{i+1}}$ are respectively the input and output vectors of the neurons in layer $i$, $\b^{(i)}\in\mR^{d_{i+1}}$ are the vectors of the so-called neuron threshold values, and
$\f^{(i)}(\cdot)=[\f_1^{(i)}(\cdot),\f_2^{(i)}(\cdot),\dots,\f_{d_{i+1}}^{(i)}(\cdot)]^T$ are $d_{i+1}\times 1$ column vectors of functions $\f_{j}^{(i)}(\cdot):\mR^{d_i}\rightarrow \mR$ that describe how $d_{i+1}$ neurons in layer $(i+1)$ operate. We also follow the standard conventions and define network's input and output as:
\begin{equation}\label{eq:model2}
  \mbox{\textbf{network input:}} \triangleq \x^{(1)}  \qquad    \mbox{\textbf{network output:}} \triangleq \x^{(L+1)}.
\end{equation}
We refer to the above architecture as $A(\d;\f^{(i)})$ where $\d=[d_1,d_2,\dots,d_{L}]$ (or for short $A(\d;\f)$ when all the neurons after the first layer have identical functions $\f$).

As we have already mentioned, the above is the general fully connected network architecture. Of our particular interest is the \emph{treelike} one as well. In such an architecture, one has matrices $W^{(i)}$ as sparse objects with a particular type of sparsity. Namely, assuming that $\delta_i=\frac{d_i}{d_{i+1}}$ is a positive  integer, the nonzero elements of the $j$-th (for $j\in\{1,2,\dots,d_i\}$) row of $W^{(i)}$ are restricted to columns $\{(j-1)\delta_i+1,(j-1)\delta_i+2,\dots,j\delta_i\}$.

\vspace{.1in}
\noindent \textbf{Memory capacity:} Any neuron has, as one of its most fundamental properties, ability to properly characterize certain features of a given data set. For example, assume that one is given a collection of images showing various sporting events ranging from basketball, football, and  soccer to tennis, golf, and chess. One can then wonder if an NN architecture can properly memorize which image shows which of the sports. The answer is yes (in fact, not only can a properly trained NN do so, but a single neuron can as well). To see how this question fits into the above mathematical NN frame, let us assume that we are given $m$ data pairs $(\x^{(0,k)},\y^{(0,k)})$, $k\in\{1,2,\dots,m\}$. For each $\x^{(0,k)}\in \mR^{n}$ that represents an $n$-dimensional image, $\y^{(0,k)}\in\mR$ represents its a corresponding label that explains which sport the image shows. One then observes that determining weights $W^{(i)}$ in (\ref{eq:model1}) such that
\begin{equation}\label{eq:model3}
\x^{(1)}=\x^{(0,k)}\quad \Longrightarrow \quad \x^{(L+1)}=\y^{(0,k)} \qquad \forall k,
\end{equation}
is sufficient to properly memorize the given data pairs. While the above NN concept has the ability to memorize, it doesn't necessarily mean that it can do so for any data set. As intuitively expected, it typically turns out that the properties of the underlying data set, and in particular its size, play a key role in distinguishing between when the memorization can and cannot be done. In fact, the so-called \emph{memory capacity}, $C(A)$, of a given neural architecture is exactly the largest $m$ such that (\ref{eq:model3}) holds for any collection of data pairs $(\x^{(0,k)},\y^{(0,k)})$, $k\in\{1,2,\dots,m\}$ with certain prescribed properties. Clearly, determining an, as precise as possible, value of the memory capacity, is of the fundamental importance when it comes to ensuring a proper network utilization. We below present a collection of results that directly relate to the memory capacity.

Before proceeding with the discussion of the main results, we find it useful to clearly state all the underlying assumptions that make the analysis easier to present.

\vspace{.1in}
\noindent \textbf{Structural (network architecture) assumptions:} \emph{\textbf{(i)}} We consider zero-thresholds ($\b^{(i)}=0$) \emph{committee machine} networks with one hidden layer. That means that, in the above generic setup, we have $L=3$, $W^{(1)}=I_{n\times n}$ and $W^{(3)}=\1_{d_{2}\times 1}^T$ (i.e. $W^{(3)}$ is a $d_2$-dimensional row vector of all ones). \textbf{\emph{(ii)}} For the very first layer we assume, as is done traditionally, $\f^{(1)}(W^{(1)}\x^{(1)})=\f^{(1)}(\x^{(1)})=\x^{(1)}$.  \emph{\textbf{(iii)}} For the second and third layer we consider the so-called \emph{sign} perceptrons (with zero thresholds), which implies $\f^{(i)}_j(W^{(i)}\x^{(i)}-\b^{(i)})=\mbox{sign} \left ( W^{(i)}\x^{(i)}\right )$. \emph{\textbf{(iv)}} Since the only changeable number of neurons is in the hidden layer, we define $d\triangleq d_2$ and $\delta\triangleq \delta_1=\frac{d_1}{d_2}=\frac{n}{d}$ and consequently have the following architecture $A([n,d,1];\mbox{sign})$. \emph{\textbf{(v)}} Finally, to ensure that the node in the last layer functions properly, we assume \emph{odd} values for $d$ (this can easily be modified through different definitions of $W^{(3)}$ to account for any (including even) $d$; however, the writing becomes more complicated and the overall presentation would lose on elegance). The above specified architecture for general full $W^{(i)}$'s corresponds to the so-called \emph{fully connected committee machines} (FCM). When $W^{(i)}$'s are with the above mentioned particular sparse structure, we have the so-called \emph{treelike committee machines} (TCM). Moreover, modulo the first layer, which is typically included for the completeness and conventional reasons (and as such it does not really affect the functioning of the network), the entire structure consists of \emph{sign} perceptrons. We refer to such structures as \emph{sign} perceptrons neural nets (SPNN). Depending on the type of $W^{(i)}$'s, we then have FCM or TCM SPNNs. As is well known, and as we will see later on, these objects belong to the class of \emph{discrete} neuronal functions and are very difficult to analytically handle.

\vspace{.1in}

\noindent \textbf{Technical (data related) assumptions:} \emph{\textbf{(i)}} As one expects based on the above choice of $\f$'s, we assume \emph{binary} labeling, i.e. we assume $\y_i^{(0,k)}\in\{-1,1\}$ (the entire analysis directly extends to any other form of $q$-ary labeling at the expense of writing and overall exposition being substantially more involved). \emph{\textbf{(ii)}} We assume separable data sets (for example, indistinguishable/contradictory pairs (or subgroups) like $(\x^{(0,k)},\y^{(0,k)})$ and $(\x^{(0,k)},-\y^{(0,k)})$ can not appear). \emph{\textbf{(iii)}} We assume statistical data sets and, for the easiness of the presentation, we take $\x^{(0,k)}$ as iid standard normals. While this follows the standard statistical trend from the classical single perceptron references \cite{DTbern,Gar88,StojnicGardGen13,Cover65,Winder,Winder61,Wendel62}, we should add that, when one is interested in providing capacity upper bounds that hold universally, it is sufficient to choose a particular type of the data set for which these bounds are valid. Of course, as shown in \cite{DTbern,Gar88,StojnicGardGen13,Cover65,Winder,Winder61,Wendel62} for single perceptrons, and as we will see later on for the networks, for particular statistics the bar can be raised much higher and the capacity might potentially be determined on a very precise level.

\vspace{.1in}

\noindent \textbf{Relevant prior work:} Studying the \emph{sign} perceptron neurons has been around for as long as the pattern recognition and machine learning fields  have. In fact, when repositioned within the context of integral geometry, various related studies can be found in the works that stretch back even further (see, e.g. \cite{Schlafli,Cover65,Wendel62,Joseph60}). As is by now well known, the most prominent result related to the \emph{sign} spherical perceptron, is that its capacity grows linearly with the dimensionality of the data ambient space, $n$. Moreover, one can actually precisely determine the underlying linear proportionality as $C(A(1;\mbox{sign}))\rightarrow 2n$ as $n\rightarrow\infty$. After the initial considerations from  \cite{Schlafli,Cover65,Winder,Winder61,Wendel62,Cameron60,Joseph60}, this celebrated result has been rediscovered in various forms on a multitude of occasions in various different fields \cite{BalVen87,Ven86,DT,StojnicISIT2010binary,DonTan09Univ,DTbern} with particulary interesting being statistical physics considerations from \cite{Gar88} and their rigorous justifications from \cite{StojnicGardGen13,StojnicGardSphErr13}.

\underline{\emph{Multi-layer NN (FCM versus TCM):}} While all the above mentioned results pretty much completely clarified the way a single \emph{sign} perceptron functions, the picture is not remotely as clear when one moves to the corresponding multi-perceptron counterparts. The results related to the TCM architectures are particularly scarce. A little bit more is known regarding the behavior of the related FCM ones. The direct connection between the two is not apparent though. Beyond the obvious fact that the FCM capacities trivially upper-bound the corresponding TCM ones, one may also take (in particular for the 1-hidden layer architectures) that the TCM capacities are roughly speaking the FCM ones scaled by $d$. Since these connections are rather interesting, it is useful to recall on some of the known FCM results. To that end, let the total number of the free network parameters (weights) be $w=\sum_{i=1}^{L-1} d_id_{i+1}$. Then, probably, the most closely related result to our own is the VC-dimension \emph{qualitative} memory capacity upper bound $O(w\log(w))$ (for the 1-hidden layer NNs, one has that $w=d_1d_{2}+d_2=(n+1)d$ for FCM and $w=d_1+d_{2}=n+d$ for TCM which, for huge $n$ and large $d_i$'s, gives the above mentioned FCM versus TCM capacity scaling by $d$). Although not directly related, we mention a couple of results regarding the corresponding lower bounds. First in \cite{Baum88}, it was argued that for a shallow 3-layer network (similar to the one that we study in this paper) the capacity scales as $O(nd)$. In \cite{Vershynin20}, it was shown for networks with more than three layers that the capacity is (roughly speaking) at least $O(w)$. All of these results are of the qualitative/scaling type and as such are not directly related or comparable to the ones that we will present below. In particular, we analyze the capacity for \emph{any} given $d$ and do so on the phase transition level of precision, where the qualitative/scaling type of estimates are not allowed. Nonetheless, we included all the known scaling estimates since they provide a simple descriptive characterizations as to what one might expect from the network capacities.

\underline{\emph{Different activations:}} Among relatively simple traditionally employed neuronal functions, the \emph{sign} perceptrons (of our interest in this paper) are probably the hardest to analytically handle. The main reason is their discreteness. If one, for example, deviates a little bit and allows for their various well known continuous counterparts/relaxations (i.e., for $\f$'s being sigmoid, ReLU, tanh and so on) a little bit more is known. For example, for sigmoid, it was suggested for deep nets in \cite{Yama93} (and proven for 4-layer nets in \cite{GBHuang03}) that the capacity is at least $O(w)$. Similar results were shown in \cite{ZBHRV17,HardrtMa16} with an additional restriction on the number of nodes that was later on removed in \cite{YunSuJad19}.

\underline{\emph{Algorithmic achievability:}} When it comes to the continuous neuronal functions, a whole lot of great work has been done recently on equally important problem of actually approaching the capacity as closely as possible. After an overwhelming empirical evidence suggested that simple gradient based methods might be able to train the net rather well, the focus in recent years has been on showing that the so-called mild over-parametrization (moderately larger number of all free parameters, $w$, compared to the size of the memorizable data set, $m$) suffices to justify its performance. A subset of very interesting and successful approaches towards rigorously establishing these statements can be found in e.g. \cite{DuZhaiPoc18,GeWangZhao19,ADHLW19,JiTel19,LiLiang18,OymSol19,RuoyuSun19,SongYang19,ZCZG18}. A majority of these works treats FCMs but are extendable to TCMs as well.

\underline{\emph{Statistical physics (replica methods):}} Finally, excellent results have been obtained through the statistical physics replica methods. While these are mathematically non-rigorous, they are most closely related to our work in terms of both the studied network setup and the obtained performance predictions. In particular, \cite{EKTVZ92,BHS92} study the very same, treelike (non-overlapping weights) committee machine (TCM) architecture (as well as a directly related, fully connected (overlapping weights) one). They obtained the closed form replica symmetry based capacity predictions for any number, $d$, of the neurons in the hidden layer. Moreover, \cite{EKTVZ92,BHS92} proceeded even further and studied the corresponding scaling behavior for large $d$  and showed that the obtained bound violates the uniform bound extension of \cite{Cover65,Winder,Winder61,Wendel62} given in \cite{MitchDurb89}. As a remedy to such a contradiction, they further undertook studying of the first level of the replica symmetry breaking (rsb) and showed that it can lower the capacity. Corresponding large $d$ scaling rsb considerations were presented in \cite{MonZech95} for both the committee and the so-called parity machines (for more on earlier parity machines replica considerations, see,e.g., \cite{BarKan91,BHK90}). Also, for the fully connected committee machines architecture, a bit later, \cite{Urban97,XiongKwonOh97} obtained the large $d$ scaling that matches the upper-bounding one of \cite{MitchDurb89}. More recently, \cite{BalMalZech19} obtained the first level of rsb capacity predictions for the treelike architecture but with the ReLU activations and \cite{ZavPeh21} moved things further by obtaining similar predictions for a few additional activations, including ReLU, erf, quadratic, and linear.

\vspace{.1in}
\noindent \textbf{Our contributions:} We study, within the above mentioned statistical data context, the memory capacity of TCM SPNNs with one hidden layer, i.e. we study $C(A([n,d,1];\mbox{sign}))$. The following are the main takeaways: \emph{\textbf{(i)}} We attack the problem on the \emph{\textbf{phase transition}} level, which does not allow for \emph{qualitative/scaling} type of estimates and instead requires very \emph{precise} analysis of all underlying quantities. \emph{\textbf{(ii)}} Utilizing a powerful mathematical engine, called Random Duality Theory (RDT) (see, e.g. \cite{StojnicCSetam09,StojnicICASSP10var,StojnicRegRndDlt10}), we create a generic framework for the analysis of the memory capabilities of SPNNs.
\emph{\textbf{(iii)}} For \emph{any} (odd) given number of the nodes in the hidden layer, $d$, we provide an explicit characterization of an upper bound on $\lim_{n\rightarrow\infty}C(A([n,d,1];\mbox{sign}))$. \emph{\textbf{(iv)}} Somewhat surprisingly and despite heavy underlying mathematics, the final forms of the resulting characterizations are sufficiently elegant to even allow an efficient numerical evaluation of the exact proportionality constants that the analysis produces. The obtained results are shown in Table \ref{tab:tab1}. The numerical values are the upper bounds $\hat{c}(3,\mbox{sign})$ on the true values of the $n$-scaled capacity $c(d;\mbox{sign})$ defined as
\begin{equation}\label{eq:model4}
c(d;\mbox{sign})\triangleq\lim_{n\rightarrow\infty} \frac{C(A([n,d,1];\mbox{sign}))}{n}.
\end{equation}
\emph{\textbf{(v)}} The obtained results exactly match the replica symmetry predictions of \cite{EKTVZ92,BHS92} for any $d$, thereby establishing the statistical physics predictions as rigorous upper bounds (which is in agreement and as expected by the postulates of the RDT). More, importantly, for $d\leq 5$ our results improve on the best rigorously known ones of \cite{MitchDurb89}, thereby making the first mathematically rigorous progress in over 30 years in studying this, rather fundamental, neural network open problem.

 \begin{table}[h]
  \caption{\textbf{\bl{Theoretical estimates}} of the memory capacity upper bounds of 1-hidden layer TCM SPNN}
  \label{tab:tab1}
  \centering
  \begin{tabular}{ccccc}
    \hline\hline
  \textbf{Upper bound on} & \textbf{Reference}  & \multicolumn{3}{c}{$d$}                   \\
    \cline{3-5}
    $c(d;\mbox{sign})\triangleq\lim_{n\rightarrow\infty} \frac{C(A([n,d,1];\mbox{sign}))}{n}$  &  (methodology)   & $\mathbf{1}$   & $\mathbf{3}$     & $\mathbf{5}$ \\
    \hline\hline
   $ $  \hspace{.2in} $\hat{c}(d;\mbox{sign})$ \hspace{.2in} $ $ &  this paper \bl{\textbf{(RDT)}}& \bl{$\mathbf{2}$} & \bl{$\mathbf{4.025}$}  & \bl{$\mathbf{5.769}$}  \\
    \hline
    $c_{RS}(d;\mbox{sign})$ & $ $  \hspace{.0in}   \cite{EKTVZ92,BHS92} \red{(\textbf{Replica symmetry)}} \hspace{.0in} $ $   & \red{$\mathbf{2}$} & \red{$\mathbf{4.025}$}  & \red{$\mathbf{5.769}$}      \\
     \hline
    $c_{CG}(d;\mbox{sign})$ & $ $  \hspace{.0in}   \cite{MitchDurb89} \dgr{(\textbf{Combinatorial geometry)}}  \hspace{.0in} $ $   & \dgr{$\mathbf{2}$} & \dgr{$\mathbf{5.421}$}  & \dgr{$\mathbf{6.425}$}  \\
       \hline\hline
  \end{tabular}
\end{table}

\section{Technical analysis}
\label{sec:analysis}

We start things off by putting everything on the right mathematical track. To that end we first set $W\triangleq W^{(2)}$ and after recalling that $W^{(1)}=I$ and $W^{(3)}=\1_{d\times 1}^T=\1^T$ (where we drop the subscript in $\1$ to ease the writing), we have for any $k\in\{1,2,\dots,m\}$
\begin{equation}\label{eq:ta1}
\x^{(1)}=\x^{(0,k)}  \quad \Longrightarrow \quad  \x^{(2)}=\f^{(1)}(W^{(1)}\x^{(1)})=\f^{(1)}(\x^{(1)})=\x^{(1)}=\x^{(0,k)},
\end{equation}
and
\begin{equation}\label{eq:ta2}
\x^{(2)}=\x^{(0,k)}  \quad \Longrightarrow \quad  \x^{(3)}=\f^{(2)}(W^{(2)}\x^{(2)})=\mbox{sign}(W\x^{(0,k)}),
\end{equation}
and
\begin{equation}\label{eq:ta3}
 \x^{(3)}=\mbox{sign}(W^{(2)}\x^{(0,k)}) \quad \Longrightarrow \quad  \x^{(4)}=\f^{(3)}(W^{(3)}\x^{(3)})=\mbox{sign}(\1^T\mbox{sign}(W\x^{(0,k)})).
\end{equation}
Connecting beginning in (\ref{eq:ta1}) and end in (\ref{eq:ta3}), we have an explicit, closed-form relation between the input and the output of the network
\begin{equation}\label{eq:ta4}
\x^{(1)}=\x^{(0,k)} \quad \Longrightarrow \quad  \x^{(4)}=\mbox{sign}(\1^T\mbox{sign}(W\x^{(0,k)})).
\end{equation}
The following condition is then both necessary and sufficient for the network to operate properly
\begin{equation}\label{eq:ta5}
\y^{(0,k)}=\mbox{sign}(\1^T\mbox{sign}(W\x^{(0,k)})).
\end{equation}
After setting
\begin{equation}\label{eq:ta6}
\y\triangleq \begin{bmatrix}
               \y^{(0,1)} & \y^{(0,2)} & \dots & \y^{(0,m)}
            \end{bmatrix}^T   \qquad \mbox{and} \qquad X\triangleq \begin{bmatrix}
               \x^{(0,1)} & \x^{(0,2)} & \dots & \x^{(0,m)}
            \end{bmatrix}^T,
\end{equation}
one can rewrite (\ref{eq:ta5}) as
\begin{equation}\label{eq:ta7}
\left (\exists W\in\mR^{d\times n}| \|\y^T-\mbox{sign}(\1^T\mbox{sign}(WX^T))\|_2=0 \right )  \quad \Longleftrightarrow \quad \left ( \left ( X,\y \right ) \mbox{is memorized} \right ),
\end{equation}
and the $k$-th data pair $\left ( \x^{(0,k)},\y^{(0,k)} \right )$ are the $k$-th rows of $m\times n$ matrix $X$ and $m\times 1$ column vector $\y$. We can further alternatively write (\ref{eq:ta7}) as the following
\begin{center}
 	\tcbset{beamer,sidebyside,lower separated=false, fonttitle=\bfseries, coltext=black,
		interior style={top color=yellow!20!white, bottom color=yellow!60!white},title style={left color=black, right color=red!50!blue!60!white},
		width=(\linewidth-4pt)/4,before=,after=\hfill,fonttitle=\bfseries,equal height group=AT}
 	\begin{tcolorbox}[title=Algebraic memorization characterization:,sidebyside,width=1\linewidth]
\vspace{-.15in}\begin{eqnarray}\label{eq:ta8}
\hspace{-.3in} 0=\xi\triangleq \min_{W,Q} & & \hspace{-.1in}\|\y-\mbox{sign}(\mbox{sign}(Q) \1)\|_2 \nonumber \\
\hspace{-.5in} \mbox{subject to} & & \hspace{-.1in}XW^T=Q
\end{eqnarray}
 \tcblower
 \hspace{-.2in}$\Longleftrightarrow$ \hspace{.1in} Data set $\left (X,\y \right )$ is properly memorized.
 		\vspace{-.0in}
 	\end{tcolorbox}
\end{center}
Since the above optimization problem is the key on the path towards ensuring proper network data memorization, we below analyze it in more detail.

\subsection{Upper-bounding capacity via Random Duality Theory (RDT)}
\label{sec:ubrdt}

As discussed in the introduction, we consider statistical data sets and assume that elements of $X$ are iid standard normals. Due to the consequential statistical rotational symmetry, we can then, analogously and without a loss of generality, assume that the elements of $\y$ are all equal to 1, i.e. that $\y=\1$. Then the above key optimization becomes
\begin{eqnarray}\label{eq:ta9}
\xi=\min_{Z,Q} & & \|\1-\mbox{sign}(\mbox{sign}(Q) \1)\|_2 \nonumber \\
 \mbox{subject to} & & XZ=Q,
\end{eqnarray}
where for the easiness of writing, we have introduced a cosmetic change $Z=W^T$. We consider the TCM architecture, with the sparse $Z$ such that the only nonzero elements of the $j$ column of $Z$ are in rows $\cS^{(j)}\triangleq\{(j-1)\delta+1,(j-1)\delta+2,\dots,j\delta\}$. Moreover, as the above problem is insensitive with respect to the
scaling of $Z$ or $Q$, it can be rewritten as
\begin{eqnarray}\label{eq:ta9a}
\xi=\min_{Z,Q} & & \|\1-\mbox{sign}(\mbox{sign}(Q) \1)\|_2 \nonumber \\
 \mbox{subject to} & & XZ=Q\nonumber \\
  & & \|Z_{:,j}\|_2=1 \nonumber \\
  & & \mbox{supp}(Z_{:,j})=\cS^{(j)}, 1\leq j\leq d,
\end{eqnarray}
where $\|Z_{:,j}\|_2$ is the norm of the $j$ column of $Z$. One can then trivially rewrite the above as
\begin{eqnarray}\label{eq:ta9aa0}
\xi=\min_{\z^{(j)},Q} & & \|\1-\mbox{sign}(\mbox{sign}(Q) \1)\|_2 \nonumber \\
 \mbox{subject to} & & X^{(j)}\z^{(j)}=Q_{:,j}, 1\leq j\leq d, \nonumber \\
  & & \|\z^{(j)}\|_2=1 \nonumber \\
  & & \z^{(j)}\in\mR^{\delta}, Q\in\mR^{m\times d},
\end{eqnarray}
where $X^{(j)}=X_{:,\cS^{(j)}}\in\mR^{m\times \delta}$. To attack the statistical optimization in (\ref{eq:ta9a}), we heavily rely on the powerful mathematical engine called Random Duality Theory (RDT) developed in a long series of work \cite{StojnicCSetam09,StojnicICASSP10var,StojnicCSetamBlock09,StojnicICASSP10block,StojnicRegRndDlt10}. Also, given the difficulty of the problem that we are facing, we assume a complete familiarity with the RDT and start by first summarizing its main principles and then continue by showing, step-by-step, how each of those principles applies to the problems of our interest here.

\vspace{-.0in}\begin{center}
 	\tcbset{beamer,lower separated=false, fonttitle=\bfseries, coltext=black ,
		interior style={top color=yellow!20!white, bottom color=yellow!60!white},title style={left color=black!80!purple!60!cyan, right color=yellow!80!white},
		width=(\linewidth-4pt)/4,before=,after=\hfill,fonttitle=\bfseries}
 \begin{tcolorbox}[beamer,title={\small Summary of the RDT's main principles} \cite{StojnicCSetam09,StojnicRegRndDlt10}, width=1\linewidth]
\vspace{-.15in}
{\small \begin{eqnarray*}
 \begin{array}{ll}
\hspace{-.19in} \mbox{1) \emph{Finding underlying optimization algebraic representation}}
 & \hspace{-.0in} \mbox{2) \emph{Determining the random dual}} \\
\hspace{-.19in} \mbox{3) \emph{Handling the random dual}} &
 \hspace{-.0in} \mbox{4) \emph{Double-checking strong random duality.}}
 \end{array}
  \end{eqnarray*}}
\vspace{-.2in}
 \end{tcolorbox}
\end{center}\vspace{-.0in}

To make the presentation look neat, we adopt the practice to formalize all the key results (including both simple to more complicated ones) as lemmas and theorems.

\vspace{.1in}

\noindent \underline{1) \textbf{\emph{Algebraic memorization characterization:}}}  The following lemma summarizes the above algebraic discussion which ultimately provided a convenient optimization representation of the underlying network memorization property.
\begin{lemma}(Algebraic optimization representation)
Assume a 1-hidden layer TCM SPNN with architecture $A([n,d,1];\mbox{sign})$. Any given data set $\left (\x^{(0,k)},1\right )_{k=1:m}$ can not be properly memorized by the network if
\begin{equation}\label{eq:ta10}
  f_{rp}(X)>0,
\end{equation}
where
\begin{equation}\label{eq:ta11}
f_{rp}(X)\triangleq \frac{1}{\sqrt{n}}\min_{\|\z^{(j)}\|_2=1,Q} \max_{\Lambda\in\mR^{m\times d}} \|\1-\emph{\mbox{sign}}(\emph{\mbox{sign}}(Q) \1)\|_2 +\sum_{j=1}^{d}(\Lambda_{:,j})^TX^{(j)}\z^{(j)} -\tr(\Lambda^TQ),
\end{equation}
and $X\triangleq \begin{bmatrix}
               \x^{(0,1)} & \x^{(0,2)} & \dots & \x^{(0,m)}
            \end{bmatrix}^T$.
  \label{lemma:lemma1}
\end{lemma}
\begin{proof}
One first writes the Lagrangian of the optimization in (\ref{eq:ta9}) as
\begin{equation}\label{eq:ta12}
\mathcal{L}(Z,Q,\Lambda)=\|\1-\emph{\mbox{sign}}(\emph{\mbox{sign}}(Q) \1)\|_2 +\sum_{j=1}^{d}(\Lambda_{:,j})^TX^{(j)}\z^{(j)} -\tr(\Lambda^TQ),
\end{equation}
and then observes that
\begin{eqnarray}\label{eq:ta13}
\xi & \triangleq & \min_{\|\z^{(j)}\|_2=1,Q,X^{(j)}\z^{(j)}=Q_{:,j}} \|\1-\emph{\mbox{sign}}(\emph{\mbox{sign}}(Q) \1)\|_2
= \min_{\|\z^{(j)}\|_2=1,Q}
\max_{\Lambda}\mathcal{L}(Z,Q,\Lambda) \nonumber \\
 & = & \min_{\|\z^{(j)}\|_2=1,Q} \max_{\Lambda} \|\1-\emph{\mbox{sign}}(\emph{\mbox{sign}}(Q) \1)\|_2 +\sum_{j=1}^{d}(\Lambda_{:,j})^TX^{(j)}\z^{(j)} -\tr(\Lambda^TQ) \nonumber \\
 & \triangleq & f_{rp}(X).
\end{eqnarray}
Connecting beginning and end in (\ref{eq:ta13}) gives $\xi=f_{rp}(X)$. Based on (\ref{eq:ta8}), one then has that the data can be properly memorized if and only if $f_{rp}(X)=0$. This also means that as soon as $f_{rp}(X)>0$ the network can not properly memorize the data, which is exactly what is claimed in the lemma.
\end{proof}

In what follows, we consider mathematically the most challenging, so-called \emph{linear}, regime with
\begin{equation}\label{eq:ta14}
  \alpha\triangleq \lim_{n\rightarrow\infty}\frac{m}{n}.
\end{equation}
While the above lemma holds for any given data set $\left (\x^{(0,k)},1\right )_{k=1:m}$, to analyze (\ref{eq:ta10}) and (\ref{eq:ta11}), the RDT proceeds by imposing a statistics  on $X$.


\vspace{.1in}
\noindent \underline{2) \textbf{\emph{Determining the random dual:}}} As is standard within the RDT, we also utilize the so-called concentration of measure property, which basically means that for any fixed $\epsilon >0$,  we have (see, e.g. \cite{StojnicCSetam09,StojnicRegRndDlt10,StojnicICASSP10var})
\begin{equation*}
\lim_{n\rightarrow\infty}\mP_X\left (\frac{|f_{rp}(X)-\mE_X(f_{rp}(X))|}{\mE_X(f_{rp}(X))}>\epsilon\right )\longrightarrow 0.\label{eq:ta15}
\end{equation*}
Another key ingredient of the RDT machinery is the following so-called random dual theorem.
\begin{theorem}(Memorization characterization via random dual) Let $d$ be any odd integer. Consider TCM SPNN with $d$ neurons in the hidden layer, and architecture  $A([n,d,1];\mbox{sign})$, and let the elements of $X\in\mR^{m\times n}$ , $G\in\mR^{m\times d}$, and $H\in\mR^{\delta\times d}$ be iid standard normals. Set
\vspace{-.0in}
\begin{eqnarray}
\phi(Q) & \triangleq & \|\1-\emph{\mbox{sign}}(\emph{\mbox{sign}}(Q) \1)\|_2\nonumber \\
 f_{rd}(G,H) & \triangleq &
 \frac{1}{\sqrt{n}}\min_{\phi(Q)=0,\|\z^{(j)}\|_2=1}\max_{\Lambda \in R^{m\times d},\|\Lambda\|_F=1}\lp \tr(\Lambda^T G)+\sum_{j=1}^{d}\|\Lambda_{:,j}\|_2(H_{:,j})^T\z^{(j)} -\tr(\Lambda^TQ) \rp \nonumber \\
 \phi_0 & \triangleq & \lim_{n\rightarrow\infty} \mE_{G,H}f_{rd}(G,H)  .\label{eq:ta16}
\vspace{-.0in}\end{eqnarray}
One then has \vspace{-.02in}
\begin{eqnarray}
\hspace{-.3in}(\phi_0  > 0)   &  \Longrightarrow  & \lp \lim_{n\rightarrow\infty}\mP_{X}(f_{rd}>0)\longrightarrow 1\rp
\quad  \Longrightarrow \quad \lp \lim_{n\rightarrow\infty}\mP_{X}(f_{rp}>0)\longrightarrow 1 \rp  \nonumber \\
& \Longrightarrow & \lp \lim_{n\rightarrow\infty}\mP_{X}(A([n,d,1];\mbox{sign}) \quad \mbox{fails to memorize data set} \quad (X,\1))\longrightarrow 1\rp.\label{eq:ta17}
\end{eqnarray}
\label{thm:thm1}
\end{theorem}\vspace{-.17in}
\begin{proof}
  Follows as a direct $d$-fold application of the Gordon's probabilistic comparison theorem (see, e.g., Theorem B in \cite{Gordon88} as well as Theorem 1, Corollary 1, and Section 2.7.2 in \cite{Stojnicgscomp16} and Theorem 1, Corollary 1, and Section 2.3.2 in \cite{Stojnicgscompyx16}).

  In particular, for $d=1$, the statement of the above theorem becomes exactly the Gordon's comparison inequality, i.e., it becomes exactly, say, Theorem B in \cite{Gordon88}. For $d>1$, one then observes that
  \begin{eqnarray}
  f_{rd}(G,H) & \triangleq &
 \frac{1}{\sqrt{n}}\min_{\phi(Q)=0,\|\z^{(j)}\|_2=1}\max_{\|\Lambda\|_F=1}\lp \tr(\Lambda^T G)+\sum_{j=1}^{d}\|\Lambda_{:,j}\|_2(H_{:,j})^T\z^{(j)} -\tr(\Lambda^TQ) \rp \nonumber \\
 & = &
 \frac{1}{\sqrt{n}}\min_{\phi(Q)=0,\|\z^{(j)}\|_2=1}\max_{\|\Lambda\|_F=1}\lp \sum_{j=1}^{d}\Lambda_{:,j}^T G_{:,j}+\sum_{j=1}^{d}\|\Lambda_{:,j}\|_2(H_{:,j})^T\z^{(j)} -\tr(\Lambda^TQ) \rp \nonumber \\
 & = &
 \frac{1}{\sqrt{n}}\min_{\phi(Q)=0,\|\z^{(j)}\|_2=1}\max_{\|\Lambda\|_F=1}\lp \sum_{j=1}^{d}\|\z^{(j)}\|_2\Lambda_{:,j}^T G_{:,j}+\sum_{j=1}^{d}\|\Lambda_{:,j}\|_2(H_{:,j})^T\z^{(j)} -\tr(\Lambda^TQ) \rp, \nonumber \\
 \label{eq:proofthm1eq2}
\end{eqnarray}
which contains exactly the summation of the $d$ terms from the lower-bounding side of the Gordon's inequality. Given that (\ref{eq:ta11}) also contains the summation of exactly the corresponding $d$ terms from the other side of the inequality, the proof is completed.
\end{proof}
%
%
%
%
%
\vspace{.1in}
\noindent \underline{3) \textbf{\emph{Handling the random dual:}}} After solving the inner maximization over $\Lambda$ and minimization over $\z^{(j)}$, we obtain for $f_{rd}(G,H)$ from (\ref{eq:ta16})
\begin{eqnarray}
 f_{rd}(G,H) & = &   \frac{1}{\sqrt{n}}
 \min_{\phi(Q)=0,\|\Lambda\|_F=1} \sum_{j=1}^{d} \|\Lambda_{:,j}\|_2\lp \|G_{:,j}-Q_{:,j}\|_2-\|H_{:,j}\|_2 \rp \nonumber \\
 & = &   \frac{1}{\sqrt{n}}
 \min_{\phi(Q)=0}\sqrt{ \sum_{j=1}^{d} \lp \|G_{:,j}-Q_{:,j}\|_2-\|H_{:,j}\|_2 \rp^2} \nonumber \\
 & = &   \frac{1}{\sqrt{n}}
 \min_{\phi(Q)=0}\sqrt{ \|G-Q\|_F^2-2\sum_{j=1}^{d} \|G_{:,j}-Q_{:,j}\|_2\|H_{:,j}\|_2 +\|H\|_F^2}.\label{eq:ta18}
\end{eqnarray}
Then
\begin{eqnarray}
\phi_0 & = & \lim_{n\rightarrow \infty}\mE_{G,H} f_{rd}(G,H) \nonumber \\
   & = &  \lim_{n\rightarrow \infty}\mE_{G,H} \frac{1}{\sqrt{n}}
 \min_{\phi(Q)=0}\sqrt{ \|G-Q\|_F^2-2\sum_{j=1}^{d} \|G_{:,j}-Q_{:,j}\|_2\|H_{:,j}\|_2 +\|H\|_F^2} \nonumber \\
   & = &  \lim_{n\rightarrow \infty}\mE_{G,H} \frac{1}{\sqrt{n}}
 \min_{\phi(Q)=0}\sqrt{ \|G-Q\|_F^2-2d\sqrt{\frac{ \|G-Q\|_F^2}{d}}\sqrt{\delta} +\delta d} \nonumber \\
   & = &  \lim_{n\rightarrow \infty}\mE_{G,H} \frac{1}{\sqrt{n}}
 \min_{\phi(Q)=0}\lp\|G-Q\|_F- \sqrt{\delta d}\rp \nonumber \\
   & = &  \lim_{n\rightarrow \infty}\mE_{G,H} \frac{1}{\sqrt{n}}
 \min_{\phi(Q)=0} \|G-Q\|_F- 1.\label{eq:ta18a0}
\end{eqnarray}
In the above considerations, we relied on the concentrations to obtain the exact equalities. It is useful to note that the Cauchy-Schwartz enables to obtain the corresponding inequalities (which are also sufficient for the subsequent analysis) without actually relying on the concentrations. Namely,
by the Cauchy-Schwartz inequality, we first find
\begin{equation}
\sum_{j=1}^{d}  \|G_{:,j}-Q_{:,j}\|_2 \|H_{:,j})\|_2
\leq \sqrt{\sum_{j=1}^{d} \|G_{:,j}-Q_{:,j}\|_2^2} \sqrt{\sum_{j=1}^{d}\|H_{:,j})\|_2^2}
=\|G-Q\|_F\|H\|_F.
\label{eq:supp6a3}
\end{equation}
We then also easily obtain
\begin{eqnarray}
\sqrt{\|G-Q\|_F^2 -2\sum_{j=1}^{d}  \|G_{:,j}-Q_{:,j}\|_2 \|H_{:,j})\|_2 +\|H\|_F^2  }
& \geq &
\sqrt{\|G-Q\|_F^2 -2\|G-Q\|_F\|H\|_F +\|H\|_F^2  } \nonumber \\
& \geq & \|G-Q\|_F -\|H\|_F.
\label{eq:supp6a4}
\end{eqnarray}
The above would then give the following inequality type alternative to (\ref{eq:ta18a0})
\begin{eqnarray}
\phi_0 & = & \lim_{n\rightarrow \infty}\mE_{G,H} f_{rd}(G,H) \nonumber \\
   & = &  \lim_{n\rightarrow \infty}\mE_{G,H} \frac{1}{\sqrt{n}}
 \min_{\phi(Q)=0}\sqrt{ \|G-Q\|_F^2-2\sum_{j=1}^{d} \|G_{:,j}-Q_{:,j}\|_2\|H_{:,j}\|_2 +\|H\|_F^2} \nonumber \\
   & \geq &  \lim_{n\rightarrow \infty}\mE_{G,H} \frac{1}{\sqrt{n}}
 \min_{\phi(Q)=0}  \|G-Q\|_F -\|H\|_F \nonumber \\
   & = &  \lim_{n\rightarrow \infty}\mE_{G} \frac{1}{\sqrt{n}}
 \min_{\phi(Q)=0}\lp\|G-Q\|_F- \sqrt{\delta d}\rp \nonumber \\
   & = &  \lim_{n\rightarrow \infty}\mE_{G} \frac{1}{\sqrt{n}}
 \min_{\phi(Q)=0} \|G-Q\|_F- 1.\label{eq:ta18a0a0}
\end{eqnarray}

Now we focus on the remaining optimization over $Q$. To that end we set
\begin{equation}\label{eq:ta19}
  \phi_i(Q_{i,1:d})\triangleq \mbox{sign}(\mbox{sign}(Q_{i,1:d})\1),
\end{equation}
and write
\begin{eqnarray}\label{eq:ta20}
   \min_{\phi(Q)=0} \|G-Q\|_F &=& \sqrt{\min_{\phi(Q)=0} \sum_{i=1}^{m}\sum_{j=1}^{d}(G_{ij}-Q_{ij})^2}  =  \sqrt{\sum_{i=1}^{m} \min_{\phi(Q_{i,1:d})=1}  \sum_{j=1}^{d}(G_{ij}-Q_{ij})^2} \nonumber \\
   &=& \sqrt{\sum_{i=1}^{m}  z_i(G_{i,1:d})},
\end{eqnarray}
with
\begin{equation}\label{eq:ta21}
  z_i(G_{i,1:d})\triangleq \min_{\phi(Q_{i,1:d})=1}  \sum_{j=1}^{d}(G_{ij}-Q_{ij})^2.
\end{equation}
To make generic considerations easier, we first look at the simplest scenario, $d=3$.

\noindent \underline{\textbf{\emph{Handling $z_i(G_{i,1:d})$ for $d=3$:}}} When $d=3$, there are four options for signs of $Q_{i,1:3}$, i.e. $\mbox{sign}(Q_{i,1:3})\in\{[1,1,1],[-1,1,1],[1,-1,1],[1,1,-1]\}$. When signs of $G_{ij}$ and $Q_{ij}$ match the contribution of the corresponding term $(G_{ij}-Q_{ij})^2$ is zero. When there is a mismatch, the contribution is $G_{ij}^2$. Utilizing this observation, one then finds the corresponding contribution for any of the four $\mbox{sign}(Q_{i,1:3})$ and ultimately, as their minimum, the overall contribution
\begin{eqnarray}\label{eq:ta22}
  \phi_z^{(1)} & = & \max(-G_{i1},0)^2+\max(-G_{i2},0)^2+\max(-G_{i3},0)^2 \nonumber \\
  \phi_z^{(2)} & = & \max(-(-1)G_{i1},0)^2+\max(-G_{i2},0)^2+\max(-G_{i3},0)^2 \nonumber \\
  \phi_z^{(3)} & = & \max(-G_{i1},0)^2+\max(-(-1)G_{i2},0)^2+\max(-G_{i3},0)^2 \nonumber \\
  \phi_z^{(4)} & = & \max(-G_{i1},0)^2+\max(-G_{i2},0)^2+\max(-(-1)G_{i3},0)^2 \nonumber \\
  z_i(G_{i,1:3}) & = & \min_i  \phi_z^{(i)}.
\end{eqnarray}
When signs of $G_{i,1:3}$ match any of the four $\mbox{sign}(Q_{i,1:3})$ options, one clearly has $z_i(G_{i,1:3})=0$. The interesting cases therefore are when there are mismatches between the signs of $G_{i,1:3}$ and allowed signs of $Q_{i,1:3}$. Exactly four such cases exist $\{[-1,-1,-1],[1,-1,-1],[-1,1,-1],[-1,-1,1]\}$. We now separately look at the possible signs choices for $G_{i,1:3}$.

\noindent \underline{\textbf{\emph{(i) $\mbox{sign}(G_{i,1:3})=[-1,-1,-1]$:}}} From (\ref{eq:ta22}), we then have
\begin{equation}\label{eq:ta23}
  \phi_z^{(1)}  =  G_{i1}^2+G_{i2}^2+G_{i3}^2 \quad \phi_z^{(2)}  = G_{i2}^2+G_{i3}^2 \quad \phi_z^{(3)}  =  G_{i1}^2+G_{i3}^2 \quad  \phi_z^{(4)}  =  G_{i1}^2+G_{i2}^2,
 \end{equation}
and
\begin{eqnarray}\label{eq:ta24}
   z_i(G_{i,1:3})  =  \min_i  \phi_z^{(i)}= \min  \{\phi_z^{(2)},\phi_z^{(3)},\phi_z^{(4)}\}
   =\min \{G_{i2}^2+G_{i3}^2,G_{i1}^2+G_{i3}^2,G_{i1}^2+G_{i2}^2\}.
\end{eqnarray}

\noindent \underline{\textbf{\emph{(ii) $\mbox{sign}(G_{i,1:3})=[1,-1,-1]$:}}} From (\ref{eq:ta22}), we then also have
\begin{equation}\label{eq:ta25}
  \phi_z^{(1)}  = G_{i2}^2+G_{i3}^2 \quad \phi_z^{(2)}  =  G_{i1}^2+G_{i2}^2+G_{i3}^2 \quad \phi_z^{(3)}  =  G_{i3}^2 \quad  \phi_z^{(4)}  =  G_{i2}^2,
 \end{equation}
and
\begin{eqnarray}\label{eq:ta26}
   z_i(G_{i,1:3})  =  \min_i  \phi_z^{(i)}= \min  \{\phi_z^{(3)},\phi_z^{(4)}\}
   =\min \{G_{i3}^2,G_{i2}^2\}.
\end{eqnarray}
Due to symmetry, one then analogously easily finds as well
\begin{eqnarray}\label{eq:ta27}
\mbox{sign}(G_{i,1:3})=[-1,1,-1] & \Longrightarrow  & z_i(G_{i,1:3}) =\min \{G_{i3}^2,G_{i1}^2\} \nonumber \\
\mbox{sign}(G_{i,1:3})=[-1,-1,1] & \Longrightarrow  & z_i(G_{i,1:3}) =\min \{G_{i2}^2,G_{i1}^2\}.
\end{eqnarray}
 A combination of (\ref{eq:ta21}),(\ref{eq:ta21}), (\ref{eq:ta21}), (\ref{eq:ta21}) , and (\ref{eq:ta21}) and the symmetry of standard normals give
\begin{eqnarray}
  \phi_0 & \triangleq & \lim_{n\rightarrow\infty} \mE_{G,H}f_{rd}(G,H) = \lim_{n\rightarrow \infty}\mE_{G,H} \frac{1}{\sqrt{n}}
 \min_{\phi(Q)=0} \|G-Q\|_F- 1 \nonumber \\
 & = & \sqrt{\lim_{n\rightarrow\infty} \mE_{G} \frac{\sum_{i=1}^m \lp\frac{1}{8}\min \{G_{i2}^2+G_{i3}^2,G_{i1}^2+G_{i3}^2,G_{i1}^2+G_{i2}^2\}+\frac{3}{8}\min\{G_{i2}^2,G_{i1}^2\}\rp}{n}}  -1 \nonumber \\
 & = & \sqrt{\alpha\mE_{G} \lp\frac{1}{8}\min \{G_{i2}^2+G_{i3}^2,G_{i1}^2+G_{i3}^2,G_{i1}^2+G_{i2}^2\}+\frac{3}{8}\min\{G_{i2}^2,G_{i1}^2\}\rp}  -1.\label{eq:ta28}
\vspace{-.0in}\end{eqnarray}
 We summarize the above results in the following lemma.
\begin{lemma}(Memory capacity upper bound; $d=3$) Assume the setup of Theorem \ref{thm:thm1} with $d=3$. For $d=3$, consider the $n$-scaled memory capacity from (\ref{eq:model4}), $c(3;\mbox{sign})$ and the following
\vspace{-.0in}
\vspace{-.0in}\begin{center}
\tcbset{beamer,lower separated=false, fonttitle=\bfseries,
coltext=black , interior style={top color=orange!10!yellow!30!white, bottom color=yellow!80!yellow!50!white}, title style={left color=orange!10!cyan!30!blue, right color=green!70!blue!20!black}}
 \begin{tcolorbox}[beamer,title=\textbf{($n$-scaled, $d=3$) memory capacity upper bound:},lower separated=false, fonttitle=\bfseries,width=.9\linewidth] 
\vspace{-.15in}
 \begin{eqnarray*}
\hspace{-.0in} \hat{c}(3;\mbox{sign})=\frac{4}{\frac{3}{\sqrt{2\pi}}\int_{-\infty}^{\infty} \lp g^2+(1+g^2)\mbox{erf}\lp\frac{|g|}{\sqrt{2}}\rp-\frac{2e^{-\frac{|g|^2}{2}}|g|}{\sqrt{2\pi}}\rp
\mbox{erfc}\lp\frac{|g|}{\sqrt{2}}\rp e^{-\frac{g^2}{2}}dg}\approx \mathbf{4.025}. \end{eqnarray*}
\vspace{-.15in}
 \end{tcolorbox}
\end{center}\vspace{-.0in}
Then for any sample complexity $m$ such that $\alpha\triangleq \lim_{n\rightarrow\infty}\frac{m}{n}>\hat{c}(3;\mbox{sign})$
\begin{eqnarray}
 \lim_{n\rightarrow\infty}\mP_{X}(A([n,3,1];\mbox{sign}) \quad \mbox{fails to memorize data set} \quad (X,\1))\longrightarrow 1,\label{eq:ta30}
\end{eqnarray}
and
\begin{eqnarray}
 \lim_{n\rightarrow\infty}\mP_{X}(c(3,\mbox{sign})<\hat{c}(3,\mbox{sign}))\longrightarrow 1.\label{eq:ta30a}
\end{eqnarray}
\label{lemma:lemma2}
\end{lemma}\vspace{-.17in}
\begin{proof}
One first finds the expectations in (\ref{eq:ta28}) as
\begin{eqnarray}\label{eq:ta31}
  \mE_{G} \min \{G_{i1}^2,G_{i2}^2\} & = & 2\phi_g^{(1)} \nonumber \\
  \mE_{G} \min \{G_{i2}^2+G_{i3}^2,G_{i1}^2+G_{i3}^2,G_{i1}^2+G_{i2}^2\} & = & 6\phi_g^{(2)},
\end{eqnarray}
where
\begin{eqnarray}\label{eq:ta32}
\phi_g^{(1)} & = & \frac{1}{\sqrt{2\pi}}\int_{-\infty}^{\infty} g^2
\mbox{erfc}\lp\frac{|g|}{\sqrt{2}}\rp e^{-\frac{g^2}{2}}dg \nonumber \\
\phi_g^{(2)} & = & \frac{1}{\sqrt{2\pi}}\int_{-\infty}^{\infty} \lp (\mbox{erf}\lp\frac{|g|}{\sqrt{2}}\rp-\frac{2e^{-\frac{|g|^2}{2}}|g|}{\sqrt{2\pi}}
+g^2\mbox{erf}\lp\frac{|g|}{\sqrt{2}}\rp
\rp
\mbox{erfc}\lp\frac{|g|}{\sqrt{2}}\rp e^{-\frac{g^2}{2}}dg.
\end{eqnarray}
The rest follows from Theorem \ref{thm:thm1}, the above discussion, and after setting $\phi_0=0$ in (\ref{eq:ta28}).
\end{proof}
The above basically means that as long as the sample complexity $m$ is such that $m>4.025n$ (where $n$ is the ambient dimension of the data vectors) the network will not be able to memorize the data. Consequently, one has for the memory capacity of the TCM SPNN with three neurons in the hidden layer, $C(A([n,3,1];\mbox{sign}))<4.025n$.

\noindent \underline{\textbf{\emph{Handling $z_i(G_{i,1:d})$ for general (odd) $d$:}}} For general $d$ one can repeat the reasoning strategy presented above between (\ref{eq:ta21}) and (\ref{eq:ta28}). This time though, there are $2^{d-1}$ options for the signs of $Q_{i,1:d}$ and one has to be very careful to account for each of them. Eventually, one obtains the following analogue to (\ref{eq:ta28}):
 \begin{equation}
  \phi_0    = \sqrt{\alpha\frac{\sum_{l=1}^{\lceil \frac{d}{2}\rceil } \binom{d}{\lceil\frac{d}{2}\rceil -l} \varphi_1(l;d)}{2^d}} -1, \quad \mbox{where} \quad  \varphi_1(l;d)    \triangleq  \mE_{G}\left \|\bar{G}_{1:l}^{(\lfloor\frac{d}{2}\rfloor+l)}\right \|_2^2,\label{eq:ta33}
\vspace{-.0in}\end{equation}
 and $\bar{G}^{(p)}$ is the vector obtained by taking the first $p$ components of $G_{i,1:d}$ and sorting them in the increasing order of their magnitudes. After handling the expectations, one arrives at the following theorem.
\begin{theorem}(Memory capacity upper bound; general $d$) Assume the setup of Theorem \ref{thm:thm1} with general $d$. Let the network $n$-scaled capacity, $c(d,\mbox{sign})$, be as defined in (\ref{eq:model4}). First one has
 \begin{eqnarray}\label{eq:34}
 \varphi_1(l;d)=\varphi_2(l;d)
\frac{1}{\sqrt{2\pi}}\int_{-\infty}^{\infty} \lp (l-1)\varphi_3(g)\varphi_4(g)^{l-2}
+g^2\varphi_4(g)^{l-1}\rp
\lp 1-\varphi_4(g)\rp^{\lfloor\frac{d}{2}\rfloor} e^{-\frac{g^2}{2}}dg, \vspace{-.0in}\end{eqnarray}
where
 \begin{eqnarray}\label{eq:35}
\varphi_2(l;d)=l\binom{\lfloor\frac{d}{2}\rfloor+l}{l}, \quad \varphi_3(g)= \mbox{erf}\lp\frac{|g|}{\sqrt{2}}\rp-\frac{2e^{-\frac{|g|^2}{2}}|g|}{\sqrt{2\pi}}, \quad \varphi_4(g)\triangleq\mbox{erf}\lp\frac{|g|}{\sqrt{2}}\rp.\end{eqnarray}
Further, consider the following
\vspace{-.0in}
\vspace{-.0in}\begin{center}
\tcbset{beamer,lower separated=false, fonttitle=\bfseries,
coltext=black , interior style={top color=orange!10!yellow!30!white, bottom color=yellow!80!yellow!50!white}, title style={left color=orange!10!cyan!30!blue, right color=green!70!blue!20!black}}
 \begin{tcolorbox}[beamer,title=\textbf{($n$-scaled general $d$) memory capacity upper bound:},lower separated=false, fonttitle=\bfseries,width=.6\linewidth] 
\vspace{-.15in}
\begin{eqnarray*}
\hspace{-.0in} \hat{c}(d;\mbox{sign})=\frac{1}{\frac{1}{2^d}\sum_{l=1}^{\lceil \frac{d}{2}\rceil } \binom{d}{\lceil\frac{d}{2}\rceil -l} \varphi_1(l;d)}. \end{eqnarray*}
\vspace{-.15in}
 \end{tcolorbox}
\end{center}\vspace{-.0in}
Then for any sample complexity $m$ such that $\alpha\triangleq \lim_{n\rightarrow\infty}\frac{m}{n}>\hat{c}(d;\mbox{sign})$
\begin{eqnarray}
 \lim_{n\rightarrow\infty}\mP_{X}(A([n,d,1];\mbox{sign}) \quad \mbox{fails to memorize data set} \quad (X,\1))\longrightarrow 1,\label{eq:ta36}
\end{eqnarray}
and
\begin{eqnarray}
 \lim_{n\rightarrow\infty}\mP_{X}(c(d,\mbox{sign})<\hat{c}(d,\mbox{sign}))\longrightarrow 1.\label{eq:ta37}
\end{eqnarray}
\label{theorem:thm2}
\end{theorem}\vspace{-.17in}

\begin{proof}
Clearly, there are two key components that need to be shown: \textbf{\emph{\emph{(i)}}} for any $d$, $\phi_0$ is as given in (\ref{eq:ta33}); and \textbf{\emph{\emph{(ii)}}} $\varphi_1(l;d)$ is as given in (\ref{eq:34}). We split the presentation into two parts. Each of the parts addresses one of these two key proof components.

\noindent \underline{\textbf{\emph{\emph{(i)}}} \textbf{Determining $\phi_0$ for any $d$:}} The same mechanism presented in (\ref{eq:ta22})-(\ref{eq:ta28}) for $d=3$ applies or extrapolates to any $d$. Below we formalize the extrapolation. Let $\s^{(i)}\triangleq \mbox{sign}(Q_{i,:})$ and $\g^{(i)}\triangleq \mbox{sign}(G_{i,:})$. Taking $\1$ as a column vector (of appropriate dimension) of all ones, we have for a general $d$ analogue to (\ref{eq:ta22})
   \begin{equation}\label{eq:supp18}
   z_i(G_{i,:})  = \min_i  \phi_z^{(i)} \quad \mbox{with}\quad   \phi_z^{(i)}  = \sum_{\s^{(i)}\1\geq 1}\sum_{j=1}^d\max(-s_j^{(i)}G_{ij},0)^2, 1\leq i\leq 2^{d-1}.
\end{equation}
Then, in an analogy with (\ref{eq:ta23})-(\ref{eq:ta27}), the signs of $G_{i,:}$ that bring  nonzero contributions to $  z_i(G_{i,:})$ are exactly the opposite ones of those of $Q$. In other words, the relevant $G_{i,:}$ are those with $\g^{(i)}\1\leq -1$.

\bl{\underline{($\star$) $\g^{(i)}\1=-d$:}}  Now, one observes that for $\g^{(i)}\1=-d$, (\ref{eq:supp18}) gives
   \begin{eqnarray}\label{eq:supp19}
   z_i(G_{i,:})  = \min_i  \phi_z^{(i)} \quad \mbox{with}\quad   \phi_z^{(i)}  = \sum_{\s^{(i)}\1= 1}\sum_{j=1}^d\max(-s_j^{(i)}G_{ij},0)^2, 1\leq i\leq 2^{d-1},
\end{eqnarray}
which implies
   \begin{eqnarray}\label{eq:supp20}
   z_i(G_{i,:})  = \min_{\s^{(i)}\1= 1} \frac{\1^T+\s^{(i)}}{2}(G_{i,:}\circ G_{i,:})^T=\left \|\bar{G}_{1:\lceil\frac{d}{2}\rceil}^{d}\right \|_2^2=\left \|\bar{G}_{1:\lceil\frac{d}{2}\rceil}^{(\lfloor\frac{d}{2}\rfloor+\lceil\frac{d}{2}\rceil)}\right \|_2^2,
\end{eqnarray}
where $\circ$ stands for the componentwise multiplication.

\bl{\underline{($\star\star$) $\g^{(i)}\1=-d+2$:}}  For the concreteness and without loss of generality  we take $\g^{(i)}=\begin{bmatrix}
                                                                                                             -1 & -1 & \dots & -1 & 1
                                                                                                           \end{bmatrix}$. One then observes that (\ref{eq:supp18}) gives
   \begin{equation}\label{eq:supp21}
   z_i(G_{i,:})  = \min_i  \phi_z^{(i)} \quad \mbox{with}\quad   \phi_z^{(i)}  = \sum_{\s_{1:d-1}^{(i)}\1= 0}\sum_{j=1}^d\max(-s_j^{(i)}G_{ij},0)^2, 1\leq i\leq 2^{d-1},
\end{equation}
which implies
   \begin{equation}\label{eq:supp22}
   z_i(G_{i,:})  = \min_{\s_{1:d-1}^{(i)}\1= 0} \frac{\1_{1:d-1}^T+\s_{1:d-1}^{(i)}}{2}(G_{i,:}\circ G_{i,:})^T=\left \|\bar{G}_{1:\lceil\frac{d}{2}\rceil-1}^{d-1}\right \|_2^2=\left \|\bar{G}_{1:\lceil\frac{d}{2}\rceil-1}^{(\lfloor\frac{d}{2}\rfloor+\lceil\frac{d}{2}\rceil-1)}\right \|_2^2.
\end{equation}
In a similar fashion, one then finds for $\g^{(i)}\1=-d+4$ and $\g^{(i)}=\begin{bmatrix}
                                                                                                             -1 & -1 & \dots & -1& 1 & 1
                                                                                                           \end{bmatrix}$
   \begin{equation}\label{eq:supp23}
   z_i(G_{i,:})  = \min_{\s_{1:d-2}^{(i)}\1= -1} \frac{\1_{1:d-2}^T+\s_{1:d-2}^{(i)}}{2}(G_{i,:}\circ G_{i,:})^T=\left \|\bar{G}_{1:\lceil\frac{d}{2}\rceil-2}^{d-2}\right \|_2^2=\left \|\bar{G}_{1:\lceil\frac{d}{2}\rceil-2}^{(\lfloor\frac{d}{2}\rfloor+\lceil\frac{d}{2}\rceil-2)}\right \|_2^2,
\end{equation}
and proceeds inductively until $\g^{(i)}\1=-d+(d-1)=-1$. Taking into account the associated combinatorial aspects, one finally has
 \begin{equation}
     \mE z_i(G_{i,:})= \frac{\sum_{l=1}^{\lceil \frac{d}{2}\rceil } \binom{d}{\lceil\frac{d}{2}\rceil -l} \varphi_1(l;d)}{2^d}, \quad \mbox{where} \quad  \varphi_1(l;d)    \triangleq  \mE_{G}\left \|\bar{G}_{1:l}^{(\lfloor\frac{d}{2}\rfloor+l)}\right \|_2^2.\label{eq:supp24}
\vspace{-.0in}\end{equation}
A combination of concentrations, (\ref{eq:supp24}), (\ref{eq:ta16}),(\ref{eq:ta18}), and (\ref{eq:ta20}) also gives
  \begin{equation}
  \phi_0    = \sqrt{\alpha\frac{\sum_{l=1}^{\lceil \frac{d}{2}\rceil } \binom{d}{\lceil\frac{d}{2}\rceil -l} \varphi_1(l;d)}{2^d}} -1, \quad \mbox{where} \quad  \varphi_1(l;d)    \triangleq  \mE_{G}\left \|\bar{G}_{1:l}^{(\lfloor\frac{d}{2}\rfloor+l)}\right \|_2^2,\label{eq:supp25}
\vspace{-.0in}\end{equation}
which shows that $\phi_0$ is, indeed, exactly as stated in (\ref{eq:ta33}).

\noindent \underline{\textbf{\emph{\emph{(ii)}}} \textbf{Computing $\varphi_1(l;d)$:}} We take a vector of  iid  standard normals, $\g\in\mR^{\lfloor\frac{d}{2}\rfloor+l}$,  and consider $\bar{\g}$ as its sorted version, obtained by sorting the magnitudes of $\g$ in increasing order. Then assuming, without a loss of generality, that the first $l$ magnitudes of $\g$ are the smallest and  accounting for other symmetric scenarios via a combinatorial pre-factor one has
   \begin{eqnarray}\label{eq:supp26}
\varphi_1(l;d)  &  \triangleq & \mE_{\g} \bar{\g}_{1:l}^2=\binom{\lfloor\frac{d}{2}\rfloor+l}{l}\lp \frac{1}{\sqrt{2\pi}}\rp^{\lfloor\frac{d}{2}\rfloor+l}\int_{\g_{1:l}}\|\g_{1:l}\|_2^2
\int_{|\g_{l+1:\lfloor\frac{d}{2}\rfloor+l}|\geq \max(|\g_{1:l}|)}e^{-\frac{\|\g\|_2^2}{2}}d\g \nonumber \\
& = &  \binom{\lfloor\frac{d}{2}\rfloor+l}{l}\lp \frac{1}{\sqrt{2\pi}}\rp^{l}\int_{\g_{1:l}}\|\g_{1:l}\|_2^2
\lp 1-\varphi_4(\max(|\g_{1:l}|))\rp^{\lfloor\frac{d}{2}\rfloor}e^{-\frac{\|\g_{1:l}\|_2^2}{2}}d\g_{1:l}.
 \end{eqnarray}
Moreover, one can again, without a loss of generality, assume that the largest of the first $l$ smallest magnitudes of $\g$ is $\g_l$. Accounting for other options through another combinatorial pre-factor, one also has
   \begin{eqnarray}\label{eq:supp27}
\varphi_1(l;d)
& = &  \binom{\lfloor\frac{d}{2}\rfloor+l}{l}\lp \frac{1}{\sqrt{2\pi}}\rp^{l}\int_{\g_{1:l}}\|\g_{1:l}\|_2^2
\lp 1-\varphi_4(\max(|\g_{1:l}|))\rp^{\lfloor\frac{d}{2}\rfloor}e^{-\frac{\|\g_{1:l}\|_2^2}{2}}d\g_{1:l}\nonumber \\
& = &  l\binom{\lfloor\frac{d}{2}\rfloor+l}{l}\lp \frac{1}{\sqrt{2\pi}}\rp^{l}\int_{\g_{l}}
\lp 1-\varphi_4(|\g_{l}|)\rp^{\lfloor\frac{d}{2}\rfloor}e^{\frac{-\g_{l}^2}{2}}
 \int_{|\g_{1:l-1}|\leq |\g_l|}\|\g_{1:l}\|_2^2e^{-\frac{\|\g_{1:l-1}\|_2^2}{2}}d\g_{1:l-1} d\g_l \nonumber \\
& = &  l\binom{\lfloor\frac{d}{2}\rfloor+l}{l} \frac{1}{\sqrt{2\pi}}\int_{\g_{l}}
\lp 1-\varphi_4(|\g_{l}|)\rp^{\lfloor\frac{d}{2}\rfloor}e^{\frac{-\g_{l}^2}{2}} \nonumber \\
& &  \times \lp \lp \frac{1}{\sqrt{2\pi}}\rp^{l-1}\int_{|\g_{1:l-1}|\leq |\g_l|}\|\g_{1:l-1}\|_2^2e^{-\frac{\|\g_{1:l-1}\|_2^2}{2}}d\g_{1:l-1}
+\g_l^2 \varphi_4(|\g_{l}|)^{l-1}
\rp  d\g_l  \nonumber \\
& = &  l\binom{\lfloor\frac{d}{2}\rfloor+l}{l} \frac{1}{\sqrt{2\pi}}\int_{\g_{l}}
\lp 1-\varphi_4(|\g_{l}|)\rp^{\lfloor\frac{d}{2}\rfloor}e^{\frac{-\g_{l}^2}{2}} \nonumber \\
& &  \times \lp \lp \frac{(l-1)\varphi_4(|\g_{l}|)^{l-2}}{\sqrt{2\pi}}\rp\int_{|\g_{1}|\leq |\g_l|}\|\g_{1}\|_2^2e^{-\frac{\g_{1}^2}{2}}d\g_{1}
+\g_l^2 \varphi_4(|\g_{l}|)^{l-1}
\rp  d\g_l \nonumber \\
& = &  l\binom{\lfloor\frac{d}{2}\rfloor+l}{l} \frac{1}{\sqrt{2\pi}}\int_{\g_{l}}
\lp 1-\varphi_4(|\g_{l}|)\rp^{\lfloor\frac{d}{2}\rfloor}e^{\frac{-\g_{l}^2}{2}}  \lp (l-1)\varphi_4(|\g_{l}|)^{l-2}\varphi_3(|\g_l|)
+\g_l^2 \varphi_4(|\g_{l}|)^{l-1}
\rp d\g_l .\nonumber \\
\end{eqnarray}
The first combinatorial pre-factor, $\binom{\lfloor\frac{d}{2}\rfloor+l}{l}$, accounts for the number of different smallest $l$ components  locations. The second  combinatorial pre-factor, $l$, accounts for the number of different choices for the location of the largest component  among the fixed smallest $l$ ones. As (\ref{eq:supp27}) matches (\ref{eq:34}), this completes the second portion of the proof related to computing $\varphi_1(l;d)$. Together with the first part, related to determining $\phi_0$, this then also completes the entire proof of Theorem \ref{theorem:thm2}.
\end{proof}

Despite heavy mathematics involved, the above theorem ultimately allows one to actually obtain the concrete values for $\hat{c}(d,\mbox{sign})$ for \emph{any} $d$. For the first few admissible $d$'s, the capacity upper bounds that the above theorem produces are shown explicitly in Table \ref{tab:tab1} together with the replica symmetry based predictions of \cite{EKTVZ92,BHS92} and the best known rigorous upper bounds of \cite{MitchDurb89}. The obtained results \emph{exactly} match the replica symmetry based predictions for any $d$ and for $d\leq 5$ make the very first, mathematically rigorous, progress in over 30 years. Also, as the results suggest, one may expect that adding more neurons can substantially improve the SPNN memorizing abilities.

\vspace{.1in}
 \noindent \underline{4) \textbf{\emph{Double checking the strong random duality:}}}  Finally, the last step of the RDT machinery assumes double checking if the strong random duality holds as well. Due to highly discrete nature of the \emph{sign} functions, the corresponding reversal considerations from \cite{StojnicRegRndDlt10} can not be applied and the strong random duality is not in place. This implies that the presented results are indeed strict memory capacity upper bounds.


\section{Conclusion}
\label{sec:conc}

In this paper we studied the memory capabilities of \emph{sign} perceptron neural networks (SPNNs) with treelike committee machine (TCM) architectures. While the abilities of the single perceptron neurons have been very \emph{precisely} characterized through a host of now classical works of pattern recognition, machine learning theory, high dimensional geometry, statistics, and information theory (see, e.g., \cite{Schlafli,Cover65,Winder,Winder61,Wendel62,Cameron60,Joseph60,BalVen87,Ven86,DT,StojnicISIT2010binary,DonTan09Univ,DTbern,Gar88,StojnicGardGen13,StojnicGardSphErr13}),
the corresponding network related results are much harder to obtain and are very scarce to find in the literature even on a \emph{qualitative/scaling} level. We here, in a way, bridge the gap and provide the first mathematically rigorous progress in over 30 years.

In particular, we show that one can utilize the Random Duality Theory (RDT) as a powerful mathematical engine to attack these types of problems. In particular, in a statistical context, we design a generic framework for a theoretical analysis of the memorization abilities of TCM SPNNs and demonstrate that the framework allows for the underlying analysis to be done on a \emph{\textbf{phase transition}} precision level for \textbf{\emph{any}} given (odd) number of the neurons in the hidden layer. Although the underlying mathematics is quite involved, one is eventually (perhaps surprisingly) left with a single integral functional that allows to obtain concrete numerical estimates for the memory capacity upper bounds. Moreover, the proven bounds precisely match the results obtained through the statistical physics replica theory analysis in \cite{EKTVZ92,BHS92}. For the number of hidden layer neurons $d\leq 5$, our results actually improve on the best known mathematically rigorous ones from \cite{MitchDurb89}, thereby making a first progress in studying this fundamental open NN problem in well over 30 years.

The range of possible extensions is rather wide. First, we here relied on the basic RDT principles. The more advanced ones related to partially lifted (see, e.g., \cite{StojnicGardSphNeg13,StojnicMoreSophHopBnds10}) and fully lifted RDTs can be applied as well (see, e.g., \cite{Stojnicflrdt23,Stojnichopflrdt23,Stojnicbinperflrdt23}). Second, we considered the \emph{sign} activation functions. Many others are of interest as well, ReLU, erf, quadratic, to name a few. The whole machinery can be applied to all of these as well. Finally, we considered the treelike committee machine (TCM) architecture. Various other architectures are possible and have been of interest throughout the  literature.  For example, the first next logical extension would be to the fully connected committee machines as well as various forms of parity machines. Again, the whole framework applies to these scenarios as well. Each of these specific problems will be discussed in separate papers.



\begin{singlespace}
\bibliographystyle{plain}
\bibliography{nflgscompyxRefs}

\begin{thebibliography}{10}

\bibitem{ADHLW19}
S.~Arora, S.~S. Du, W.~Hu, Z.~Li, and R.~Wang.
\newblock Fine-grained analysis of optimiza\-tion and generalization for
  overparameterized two-layer neural networks.
\newblock 2019.
\newblock available online at
  {\small\bl{\url{http://arxiv.org/abs/1901.08584}}}.

\bibitem{BalMalZech19}
C.~Baldassi, E.~M. Malatesta, and R.~Zecchina.
\newblock Properties of the geometry of solutions and capacity of multilayer
  neural networks with rectified linear unit activations.
\newblock {\em Phys. Rev. Lett.}, 123:170602, October 2019.

\bibitem{BalVen87}
P.~Baldi and S.~Venkatesh.
\newblock Number od stable points for spin-glasses and neural networks of
  higher orders.
\newblock {\em Phys. Rev. Letters}, 58(9):913--916, Mar. 1987.

\bibitem{BHK90}
E.~Barkai, D.~Hansel, and I.~Kanter.
\newblock Statistical mechanics of a multilayered neural network.
\newblock {\em Phys. Rev. Lett.}, 65(18):2312--2315, Oct 1990.

\bibitem{BHS92}
E.~Barkai, D.~Hansel, and H.~Sompolinsky.
\newblock Broken symmetries in multilayered perceptrons.
\newblock {\em Phys. Rev. A}, 45(6):4146, March 1992.

\bibitem{BarKan91}
E.~Barkai and I.~Kanter.
\newblock Storage capacity of a multilayer neural network with binary weights.
\newblock {\em Europhys. Lett.}, 14(2):107, 1991.

\bibitem{Baum88}
E.~B. Baum.
\newblock On the capabilities of multilayer perceptrons.
\newblock {\em Journal of complexity}, 4(3):193--215, 1988.

\bibitem{Cameron60}
S.~H. Cameron.
\newblock Tech-report 60-600.
\newblock {\em Proceedings of the bionics symposium}, pages 197--212, 1960.
\newblock Wright air development division, {D}ayton, {O}hio.

\bibitem{Cover65}
T.~Cover.
\newblock Geomretrical and statistical properties of systems of linear
  inequalities with applications in pattern recognition.
\newblock {\em IEEE Transactions on Electronic Computers}, (EC-14):326--334,
  1965.

\bibitem{DT}
D.~Donoho and J.~Tanner.
\newblock Neighborliness of randomly-projected simplices in high dimensions.
\newblock {\em Proc. National Academy of Sciences}, 102(27):9452--9457, 2005.

\bibitem{DonTan09Univ}
D.~Donoho and J.~Tanner.
\newblock Observed universality of phase transitions in high-dimensional
  geometry, with implications for modern data analysis and signal processing.
\newblock {\em Phylosophical transactions of the royal society A: mathematical,
  physical and engineering sciences}, 367, November 2009.

\bibitem{DTbern}
D.~Donoho and J.~Tanner.
\newblock Counting the face of randomly projected hypercubes and orthants, with
  application.
\newblock {\em Discrete and Computational Geometry}, 43:522--541, 2010.

\bibitem{DuZhaiPoc18}
S.~S. Du, X.~Zhai, B.~Poczos, and A.~Singh.
\newblock Gradient descent provably optimizes over\-parameterized neural
  networks.
\newblock 2018.
\newblock available online at
  {\small\bl{\url{http://arxiv.org/abs/1810.02054}}}.

\bibitem{EKTVZ92}
A.~Engel, H.~M. Kohler, F.~Tschepke, H.~Vollmayr, and A.~Zippelius.
\newblock Storage capacity and learning algorithms for two-layer neural
  networks.
\newblock {\em Phys. Rev. A}, 45(10):7590, May 1992.

\bibitem{MitchDurb89}
R.~M.~Durbin G.~J.~Mitchison.
\newblock Bounds on the learning capacity of some multi-layer networks.
\newblock {\em Biological Cybernetics}, 60:345--365, 1989.

\bibitem{Gar88}
E.~Gardner.
\newblock The space of interactions in neural networks models.
\newblock {\em J. Phys. A: Math. Gen.}, 21:257--270, 1988.

\bibitem{GeWangZhao19}
R.~Ge, R.~Wang, and H.~Zhao.
\newblock Mildly overparametrized neural nets can memorize training data
  efficiently.
\newblock 2019.
\newblock available online at
  {\small\bl{\url{http://arxiv.org/abs/1909.11837}}}.

\bibitem{Gordon88}
Y.~Gordon.
\newblock On {M}ilman's inequality and random subspaces which escape through a
  mesh in ${R}^n$.
\newblock {\em Geometric Aspect of of functional analysis, Isr. Semin. 1986-87,
  Lect. Notes Math}, 1317, 1988.

\bibitem{HardrtMa16}
M.~Hardt and T.~Ma.
\newblock Identity matters in deep learning.
\newblock 2016.
\newblock available online at
  {\small\bl{\url{http://arxiv.org/abs/1611.04231}}}.

\bibitem{GBHuang03}
G.~B. Huang.
\newblock Learning capability and storage capacity of two-hidden-layer
  feedforward networks.
\newblock {\em IEEE Transactions on Neural Networks}, 14(2):274--281, 2003.

\bibitem{JiTel19}
Z.~Ji and M.~Telgarsky.
\newblock Polylogarithmic width suffices for gradient descent to achieve
  arbitrarily small test error with shallow relu networks.
\newblock 2019.
\newblock available online at
  {\small\bl{\url{http://arxiv.org/abs/1909.12292}}}.

\bibitem{Joseph60}
R.~D. Joseph.
\newblock The number of orthants in $n$-space instersected by an
  $s$-dimensional subspace.
\newblock {\em Tech. memo 8, project {PARA}}, 1960.
\newblock Cornel aeronautical lab., Buffalo, N.Y.

\bibitem{LiLiang18}
Y.~Li and Y.~Liang.
\newblock Learning overparameterized neural networks via stochastic gradient
  descent on structured data.
\newblock {\em In Advances in Neural Information Processing Systems}, pages
  8157--8166, 2018.

\bibitem{MonZech95}
R.~Monasson and R.~Zecchina.
\newblock Weight space structure and internal representations: A direct
  approach to learning and generalization in multilayer neural networks.
\newblock {\em Phys. Rev. Lett.}, 75:2432, September 1995.

\bibitem{OymSol19}
S.~Oymak and M.~Soltanolkotabi.
\newblock Towards moderate overparameterization: global convergence guarantees
  for training shallow neural networks.
\newblock 2019.
\newblock available online at
  {\small\bl{\url{http://arxiv.org/abs/1902.04674}}}.

\bibitem{Schlafli}
L.~Schlafli.
\newblock {\em Gesammelte Mathematische AbhandLungen I}.
\newblock Basel, Switzerland: Verlag Birkhauser, 1950.

\bibitem{SongYang19}
Z.~Song and X.~Yang.
\newblock Quadratic suffices for over-parametrization via matrix {C}hernoff
  bound.
\newblock 2019.
\newblock available online at
  {\small\bl{\url{http://arxiv.org/abs/1906.03593}}}.

\bibitem{StojnicCSetamBlock09}
M.~Stojnic.
\newblock Block-length dependent thresholds in block-sparse compressed sensing.
\newblock available online at \bl{\url{http://arxiv.org/abs/0907.3679}}.

\bibitem{StojnicCSetam09}
M.~Stojnic.
\newblock Various thresholds for $\ell_1$-optimization in compressed sensing.
\newblock available online at \bl{\url{http://arxiv.org/abs/0907.3666}}.

\bibitem{StojnicICASSP10block}
M.~Stojnic.
\newblock Block-length dependent thresholds for $\ell_2/\ell_1$-optimization in
  block-sparse compressed sensing.
\newblock {\em ICASSP, IEEE International Conference on Acoustics, Signal and
  Speech Processing}, pages 3918--3921, 14-19 March 2010.
\newblock Dallas, TX.

\bibitem{StojnicICASSP10var}
M.~Stojnic.
\newblock $\ell_1$ optimization and its various thresholds in compressed
  sensing.
\newblock {\em ICASSP, IEEE International Conference on Acoustics, Signal and
  Speech Processing}, pages 3910--3913, 14-19 March 2010.
\newblock Dallas, TX.

\bibitem{StojnicISIT2010binary}
M.~Stojnic.
\newblock Recovery thresholds for $\ell_1$ optimization in binary compressed
  sensing.
\newblock {\em ISIT, IEEE International Symposium on Information Theory}, pages
  1593 -- 1597, 13-18 June 2010.
\newblock Austin, TX.

\bibitem{StojnicGardGen13}
M.~Stojnic.
\newblock Another look at the {G}ardner problem.
\newblock 2013.
\newblock available online at \bl{\url{http://arxiv.org/abs/1306.3979}}.

\bibitem{StojnicMoreSophHopBnds10}
M.~Stojnic.
\newblock Lifting/lowering {H}opfield models ground state energies.
\newblock 2013.
\newblock available online at \bl{\url{http://arxiv.org/abs/1306.3975}}.

\bibitem{StojnicGardSphNeg13}
M.~Stojnic.
\newblock Negative spherical perceptron.
\newblock 2013.
\newblock available online at \bl{\url{http://arxiv.org/abs/1306.3980}}.

\bibitem{StojnicRegRndDlt10}
M.~Stojnic.
\newblock Regularly random duality.
\newblock 2013.
\newblock available online at \bl{\url{http://arxiv.org/abs/1303.7295}}.

\bibitem{StojnicGardSphErr13}
M.~Stojnic.
\newblock Spherical perceptron as a storage memory with limited errors.
\newblock 2013.
\newblock available online at \bl{\url{http://arxiv.org/abs/1306.3809}}.

\bibitem{Stojnicgscompyx16}
M.~Stojnic.
\newblock Fully bilinear generic and lifted random processes comparisons.
\newblock 2016.
\newblock available online at \bl{\url{http://arxiv.org/abs/1612.08516}}.

\bibitem{Stojnicgscomp16}
M.~Stojnic.
\newblock Generic and lifted probabilistic comparisons -- max replaces minmax.
\newblock 2016.
\newblock available online at \bl{\url{http://arxiv.org/abs/1612.08506}}.

\bibitem{Stojnicbinperflrdt23}
M.~Stojnic.
\newblock Binary perceptrons capacity via fully lifted random duality theory.
\newblock 2023.
\newblock available online at arxiv.

\bibitem{Stojnicflrdt23}
M.~Stojnic.
\newblock Fully lifted random duality theory.
\newblock 2023.
\newblock available online at arxiv.

\bibitem{Stojnichopflrdt23}
M.~Stojnic.
\newblock Studying {H}opfield models via fully lifted random duality theory.
\newblock 2023.
\newblock available online at arxiv.

\bibitem{RuoyuSun19}
R.~Sun.
\newblock Optimization for deep learning: theory and algorithms.
\newblock 2019.
\newblock available online at
  {\small\bl{\url{http://arxiv.org/abs/1912.08957}}}.

\bibitem{Urban97}
R~Urbanczik.
\newblock Storage capacity of the fully-connected committee machine.
\newblock {\em J. Phys. A: Math. Gen.}, 30, 1997.

\bibitem{Ven86}
S.~Venkatesh.
\newblock Epsilon capacity of neural networks.
\newblock {\em Proc. Conf. on Neural Networks for Computing, Snowbird, UT},
  1986.

\bibitem{Vershynin20}
R.~Vershynin.
\newblock Memory capacity of neural networks with threshold and {R}e{LU}
  activations.
\newblock 2019.
\newblock available online at
  {\small\bl{\url{http://arxiv.org/abs/2001.06938}}}.

\bibitem{Wendel62}
J.~G. Wendel.
\newblock A problem in geometric probability.
\newblock {\em Mathematica Scandinavica}, 1:109--111, 1962.

\bibitem{Winder61}
R.~O. Winder.
\newblock Single stage threshold logic.
\newblock {\em Switching circuit theory and logical design}, pages 321--332,
  Sep. 1961.
\newblock AIEE Special publications S-134.

\bibitem{Winder}
R.~O. Winder.
\newblock {\em Threshold logic}.
\newblock Ph. D. dissertation, Princetoin University, 1962.

\bibitem{XiongKwonOh97}
Y.~Xiong and J.~H.~Oh C.~Kwon.
\newblock The storage capacity of a fully-connected committee machine.
\newblock {\em NIPS}, 1997.

\bibitem{Yama93}
M.~Yamasaki.
\newblock The lower bound of the capacity for a neural network with multiple
  hidden layers.
\newblock {\em In International Conference on Artificial Neural Networks},
  pages 546--549, 1993.

\bibitem{YunSuJad19}
C.~Yun, S.~Sra, and A.~Jadbabaie.
\newblock Small relu networks are powerful memorizers: a tight analysis of
  memorization capacity.
\newblock {\em In Advances in Neural Information Processing Systems}, pages
  15532--15543, 2019.

\bibitem{ZavPeh21}
J.~A. Zavatone-Veth and C.~Pehlevan.
\newblock Activation function dependence of the storage capacity of treelike
  neural networks.
\newblock {\em Phys. Rev. E}, 103:L020301, February 2021.

\bibitem{ZBHRV17}
C.~Zhang, S.~Bengio, M.~Hardt, B.~Recht, and O.~Vinyals.
\newblock Understanding deep learning requires rethinking generalization.
\newblock {\em ICLR}, 2017.

\bibitem{ZCZG18}
D.~Zou, Y.~Cao, D.~Zhou, and Q.~Gu.
\newblock Stochastic gradient descent optimizes over\-parameterized deep relu
  networks.
\newblock 2018.
\newblock available online at
  {\small\bl{\url{http://arxiv.org/abs/1811.08888}}}.

\end{thebibliography}
\end{singlespace}

\end{document}